\documentclass[sigconf]{acmart}
\AtBeginDocument{%
  }

\setcopyright{acmlicensed}
\copyrightyear{2026}
\acmYear{2026}
\acmDOI{XXXXXXX.XXXXXXX}

\acmJournal{JDS}
\acmVolume{37}
\acmNumber{4}
\acmMonth{8}

\usepackage[ruled,vlined]{algorithm2e}

\usepackage{tikz}

\copyrightyear{2026}
\acmYear{2026}
\setcopyright{cc}
\setcctype{by}
\acmConference[ISSAC '26]{51st International Symposium on Symbolic and
Algebraic Computation}{July 13--17, 2026}{Oldenburg, Germany}
\acmBooktitle{51st International Symposium on Symbolic and Algebraic
Computation (ISSAC '26), July 13--17, 2026, Oldenburg, Germany}
\acmDOI{10.1145/3815436.3815448}
\acmISBN{979-8-4007-2595-1/2026/07}

\begin{document}

\title[Learning unions of intersecting affine modules]{Learning Unions of Intersecting Affine Modules in One
Dimension with Queries}

\author{Eva González}
\orcid{0009-0009-3839-4897}
%\email{webmaster@marysville-ohio.com}
\affiliation{%
  \institution{RPTU Kaiserslautern-Landau}
  \city{Kaiserslautern}
  %\state{Rhineland-Palatinate}
  \country{Germany}
}
\email{eva.gonzalezgarcia@rptu.de}

\author{Montserrat Hermo}
\orcid{0000-0001-5627-501X}
\affiliation{%
  \institution{University of the Basque Country}
  \city{San Sebastían}
  \country{Spain}
}
\email{montserrat.hermo@ehu.eus}

\author{Anthony Lin}
\orcid{0000-0003-4715-5096}
\affiliation{%
 \institution{RPTU Kaiserslautern-Landau and MPI-SWS}
 \city{Kaiserslautern}
 %\state{Arunachal Pradesh}
 \country{Germany}
}
\email{awlin@mpi-sws.org}

% \affiliation{%

% \affiliation{%

% \affiliation{%

% \affiliation{%

\renewcommand{\shortauthors}{E. González, M. Hermo and A. Lin}
%%
%% Article type: Research, Review, Discussion, Invited or position
\acmArticleType{Research}
%%
%% Links to code and data
% \acmCodeLink{https://github.com/borisveytsman/acmart}
% \acmDataLink{htps://zenodo.org/link}
%%
%% Authors' contribution
% \acmContributions{BT and GKMT designed the study; LT, VB, and AP
%   conducted the experiments, BR, HC, CP and JS analyzed the results,
%   JPK developed analytical predictions, all authors participated in
%   writing the manuscript.}
%%
%% Sometimes the addresses are too long to fit on the page.  In this
%% case uncomment the lines below and fill them accodingly.
%%
%% \authorsaddresses{Corresponding author: Ben Trovato,
%% \href{mailto:trovato@corporation.com}{trovato@corporation.com};
%% Institute for Clarity in Documentation, P.O. Box 1212, Dublin,
%% Ohio, USA, 43017-6221}

\begin{abstract}
We study the exact learnability of finite unions of intersecting affine modules in one dimension. An affine module is a set of the form $a+\sum_{j=1}^{s}b_j \mathbb{Z}$, where $a,b_1,\ldots,b_s\in\mathbb{N}$. We say that a set definable as a finite union of affine modules is a union of intersecting affine modules if it admits a representation in which all modules have a non-empty intersection. We show that this class is efficiently exactly learnable using equivalence and subset queries. Moreover, subset queries can be replaced with membership queries when a common element is known. Our algorithm requires at most $k\log(2|x_\ell|)+2k$ counterexamples, where $k$ is the number of affine modules in the smallest representation and $x_\ell$ is the largest counterexample. This implies polynomial-time learnability in the binary representation.
\end{abstract}


\maketitle

\section*{Introduction}

Learning theory provides formal models for how an algorithm can recover an unknown structure from observations. Two classical frameworks are the PAC model, where the goal is to obtain a hypothesis that is probably approximately correct from random examples \cite{valiant1984theory}, and the exact learning model with queries, where the learner can actively interact with a teacher through queries, such as equivalence and membership, to identify the target object precisely \cite{angluin1987learning}. Angluin~\cite{angluin1988queries} showed that efficient exact learnability with equivalence queries implies efficient PAC learnability. Moreover, the PAC model is inherently passive, as the learner only receives random labeled examples. Allowing membership queries leads to an active extension of the PAC model, usually referred to as PAC learning with membership queries. While learning theory has traditionally focused on concept classes such as automata, Boolean formulas, or half-spaces, there has also been work on learning semilinear sets and specific subclasses, such as finite unions of modules or hypercubes.

Semilinear sets are exactly the sets definable in linear integer arithmetic \cite{chistikov2024introduction}, and they also coincide with the Parikh images of context-free languages \cite{parikh1966context}, establishing a connection with automata. Related work on learning automata includes \cite{van2020learning} where the authors study the learnability of weighted automata over principal ideal domains, and \cite{buna2024learning}, which presents an exact learning algorithm for $\mathbb{Z}$-automata with a teacher supporting equivalence and membership queries.

The exact learnability of semilinear sets in unary representation was studied by Takada~\cite{TAKADA1992207}. Abe~\cite{abe1995characterizing} showed that PAC learning semilinear sets is polynomial in the unary encoding of numbers for dimension $d\leq{2}$ and provided several lower bounds via reductions to the problem of predicting DNF formulas. Later, Kopczynski and To~\cite{kopczynski2010parikh} showed that PAC learning semilinear sets in unary representation is polynomial for any fixed dimension. Regarding modules, Helmbold et al.~\cite{helmbold1992learning} proved that a single module is PAC-learnable in polynomial time in the binary representation for arbitrary dimensions. Abe~\cite{abe1995characterizing} further analyzed finite unions of modules, showing that they are polynomially PAC learnable in unary but as hard as learning DNF formulas in binary. A related line of work considers exact learning with membership and equivalence queries. Goldberg et al.~\cite{goldberg1994learning} proved that finite unions of hypercubes with natural corners are polynomially learnable. More recently, Markgraf et al.~\cite{markgraf2021learning} showed that finite unions of integer hypercubes are learnable with respect to the binary encoding of numbers, and applied this result to the monadic decomposition of a fragment of linear integer arithmetic. These classes can be seen as subclasses of semilinear sets, in which the periods correspond to standard basis vectors.

We show that any finite union of one-dimensional intersecting affine modules can be learned exactly using equivalence and subset queries in time polynomial in the binary representation. Moreover, subset queries can be replaced with membership queries when a common element is known. A one-dimensional affine module is a set of the form
\begin{align*}
\rho(a,\{b_1,\dots,b_s\}):=\left\{a+\sum_{j=1}^sb_j\alpha_j~\middle|~\alpha_j\in\mathbb{Z}\right\}
\end{align*}
where $s\in\mathbb{N}$ and $a,b_1,\ldots,b_s\in\mathbb{N}$. A union of affine modules can admit different representations. We say that a set definable as a finite union of affine modules is a \emph{union of intersecting affine modules} if it admits a representation in which the intersection of all modules is non-empty. We present a learning algorithm and prove its correctness and efficiency. The time complexity depends on the size of the smallest representation and on the size of the largest counterexample returned by the equivalence queries. We are not aware of any polynomial-time learning algorithm for unions of modules in the binary representation. Existing algorithms that are polynomial in the size of the smallest representation are polynomial only in the unary representation \cite{TAKADA1992207}.
While the algorithm for learning a single module \cite{helmbold1992learning} runs in polynomial time in the binary representation, its complexity also depends on the size of the largest counterexample.

\medskip

\textbf{Paper organization.} Section~\ref{sec:preliminaries} introduces the exact learning framework together with the basic preliminaries on one-dimensional affine modules and their unions. Section~\ref{sec:learning-algorithm} presents a learning algorithm for unions of affine modules and analyzes its correctness, first for a single affine module and then for unions of affine modules. Section~\ref{sec:intersecting} proves the efficiency of the algorithm for unions of intersecting affine modules and explains how to adapt it to work with membership queries instead of subset queries when a common element is known. Finally, the conclusion contains general remarks on the algorithm and its correctness proof.

\section{Preliminaries}\label{sec:preliminaries}

We denote by $\mathbb{N}$ the set of natural numbers, including $0$, and by $\mathbb{Z}$ the set of integers. For brevity, we use the term \emph{module} to refer to one-dimensional affine modules, although the term was previously used only in the case $a=0$.

\subsection{Learning framework}\label{sec:exact-learning-framework}

For a general overview on learning theory, the reader may consult \cite{kearns1994introduction}. Our \emph{instance space} is the set of integers $\mathbb{Z}$. We denote by $\mathcal{C}_M$ the \emph{concept class} of all sets of integers definable as modules. By convention, we include the empty set and individual elements of $\mathbb{Z}$ in $\mathcal{C}_M$. The class $\mathcal{C}_{U}^*$ consists of all sets of integers that are finite unions of concepts in $\mathcal{C}_M$ such that there is a representation of the union whose modules have a non-empty intersection. Let $\mathcal{C}\in\{\mathcal{C}_M,\mathcal{C}_{U}^*\}$ and $C\in\mathcal{C}$ be a target concept. A priori, our learning algorithm can make the following types of queries.
\begin{itemize}
\item \emph{equivalence:} for $H\in\mathcal{C}$, $$e(H)=\left\{
 \begin{array}{ll}
 \top,                                  &\textrm{ if }H=C,\\
 x\in(C\setminus{H})\cup(H\setminus{C}),    &\textrm{ otherwise.}
 \end{array}\right.$$
\item \emph{subset:} for $M\in\mathcal{C}$, $$s(M)=\top\textrm{ if and only if }M\subseteq{C}\textrm{ and }\bot\textrm{ otherwise.}$$
\item \emph{membership:} for $x\in{X}$, $$m(x)=\top\textrm{ if and only if }x\in{C}\textrm{ and }\bot\textrm{ otherwise.}$$
\end{itemize}

If $e(H)$ belongs to $C$, it is a \textit{positive counterexample}, otherwise it is a \textit{negative counterexample}. Since the hypothesis is always contained in the target set, our algorithm uses only positive counterexamples, which we simply call counterexamples. Note that the learning algorithm for $\mathcal{C}_M$ uses only equivalence queries. 
On the other hand, the algorithm for $\mathcal{C}_U^*$ uses both equivalence and subset queries, although the latter can be replaced with membership queries when a common element shared by all modules in the union is known. In our case, subset queries are performed on individual modules.

A concept can have different representations. We refer to a representation of minimal size as a \emph{smallest representation}. It is assumed that numbers are represented in binary and define size$(x):=1 +\lceil\log(|x|+1)\rceil$ for $x\in\mathbb{Z}$, which is $O(\log(|x|))$. The size of a representation of a union of modules is then defined as
\[
\text{size}\!\left(\cup_{i=1}^{k}\rho(a_i,\{b_{i1},\ldots,b_{is_i}\})\right)
:=\sum_{i=1}^{k}\text{size}(a_i)
+\sum_{j=1}^{s_i}\text{size}(b_{ij}).
\]

We say that a concept class $\mathcal{C}$ is \emph{efficiently exactly learnable} with equivalence, subset, and/or membership queries if there exists an algorithm using these queries that, for every $C\in\mathcal{C}$, returns an equivalent hypothesis in time polynomial in the size of the smallest representation of $C$ and the size of the largest counterexample.

\subsection{Affine modules}

A module is a set of the form 
\begin{align}\label{def:module}
\rho(a,\{b_1,\dots,b_s\}):&=\left\{a+\sum_{j=1}^sb_j\alpha_j~\middle|~\alpha_j\in\mathbb{Z}\right\} \\
&=\rho(a,\gcd(b_1,\ldots,b_s)) \notag
\end{align}
where $s\in\mathbb{N}$ and $a,b_1,\ldots,b_s\in\mathbb{N}$. We call $a$ the \emph{representative}, and each $b_j$ a \emph{period}. Since $\alpha_j\in\mathbb{Z}$, we may assume without loss of generality that all periods are positive. As $\mathbb{Z}$ is a principal ideal domain, set~(\ref{def:module}) can be generated by a single period, and it is equivalent to $\rho(a,\gcd(b_1,\ldots,b_s))$; see~\cite[Section~1.5.2]{diekert2016discrete}. In what follows, we consider only the latter representation. When the representative and period are clear from the context, we write simply $\rho$ and its variants instead of $\rho(a,b)$, for example $\rho'=\rho(a',b')$ or $\rho_i=\rho(a_i,b_i)$. We use the symbol $\sigma=\sigma(c,d)$ to denote a particular kind of module, defined below.

We introduce several fundamental lemmas and propositions concerning modules, which, although well known, are included here for completeness; see \cite{jones1998elementary,beachy2019abstract,diekert2016discrete} for related results. The elements of a module are characterized as the integers whose distance from the representative is divisible by its period.

\begin{lemma}\label{lem:element-contained-in-cartesian-congruence-class}
$x\in\rho$ if and only if $b\mid(x-a)$.
\end{lemma}
\begin{proof}
We first prove that if $x\in\rho(a,b)$ then $b\mid(x-a)$. If $x\in\rho(a,b)$ then $x=a+b\alpha$ where $\alpha\in\mathbb{Z}$. Thus, $b\mid(x-a)$. Conversely, if $b\mid(x-a)$ then $x-a=b\alpha$ where $\alpha\in\mathbb{Z}$. Hence, $x\in\rho(a,b)$.
\end{proof}

Any element of a module can be its representative. For every module $\rho(a,b)$, there exists a representative $a$ in the interval $[0,b-1]$ by the division theorem. This representative is called the \emph{residue}. When a single period is used and the representative is chosen to be the residue, the resulting representation is minimal with respect to the size defined in Section~\ref{sec:exact-learning-framework}.
\begin{lemma}\label{lem:every-element-is-a-representative}
For all $a'\in\rho(a,b)$ we have $\rho(a,b)=\rho(a',b)$.
\end{lemma}
\begin{proof}
First we prove that $\rho(a,b)\supseteq\rho(a',b)$. If $x\in\rho(a',b)$, then $x=a'+b\alpha'$. Since $a'\in\rho(a,b)$ we have that $a'=a+b\alpha$, so $x=a+b(\alpha+\alpha')$. Thus, $x\in\rho(a,b)$. Now, we prove that $\rho(a,b)\subseteq\rho(a',b)$. If $x\in\rho(a,b)$, then $x=a+b\alpha$. Since $a=a'-b\alpha'$ we have $x=a'+b(\alpha-\alpha')$. Therefore, $x\in\rho(a',b)$.
\end{proof}

A module is contained in another if and only if the period of the latter divides both the period of the former and the difference between their representatives.
\begin{proposition}\label{pro:congruence-class-contained-condition}
$\rho'\subseteq\rho$ if and only if $b\mid{b'}$ and $b\mid(a'-a)$.
\end{proposition}
\begin{proof}
We prove that if $\rho(a',b')\subseteq\rho(a,b)$ then $b\mid{b'}$ and $b\mid(a'-a)$. Since $a'\in\rho(a,b)$, by Lemma~\ref{lem:element-contained-in-cartesian-congruence-class}, we have that  $b\mid(a'-a)$ and by Lemma~\ref{lem:every-element-is-a-representative} $\rho(a,b)=\rho(a',b)$. On the other hand, there is $x\in\rho(a',b')$ such that $x-a'=b'$. Then, $b\mid{b'}$ by Lemma~\ref{lem:element-contained-in-cartesian-congruence-class} because $x\in\rho(a,b)=\rho(a',b)$. Conversely, since $b\mid(a'-a)$, by Lemma~\ref{lem:element-contained-in-cartesian-congruence-class}, we have that $a'\in\rho(a,b)$. Moreover, by Lemma~\ref{lem:every-element-is-a-representative} $\rho(a,b)=\rho(a',b)$. Now, if $x\in\rho(a',b')$ then by Lemma~\ref{lem:element-contained-in-cartesian-congruence-class} we have $b'\mid(x-a')$, and since $b\mid{b'}$ then $b\mid(x-a')$. Therefore, by Lemma~\ref{lem:element-contained-in-cartesian-congruence-class} $x\in\rho(a,b)=\rho(a',b)$.
\end{proof}

The corollary follows from Lemma~\ref{lem:element-contained-in-cartesian-congruence-class} and Proposition~\ref{pro:congruence-class-contained-condition}. 
\begin{corollary}\label{cor:membership-eqiuvalent-to-subset}
$a+b'\in\rho(a,b)$ if and only if $\rho(a,b')\subseteq\rho(a,b)$.
\end{corollary}

If a module contains two points, then it contains the module generated by these two points.
\begin{proposition}\label{pro:module-with-x-q-is-contained}
If $x,y\in\rho$, then $\rho(x,y-x)\subseteq\rho$.
\end{proposition}
\begin{proof}
By Lemma~\ref{lem:every-element-is-a-representative} we can set $x$ as the representative of $\rho(a,b)$. Since $y\in\rho(x,b)$, we have by Lemma~\ref{lem:element-contained-in-cartesian-congruence-class} that $b\mid(y-x)$. Hence, $\rho(x,y-x)\subseteq\sigma(x,b)$ by Proposition~\ref{pro:congruence-class-contained-condition}.
\end{proof}

\subsection{Union of affine modules}

We use letter $U$ to represent a finite union of modules and use the subindex $i$ to denote a module $\rho_i=\rho(a_i,b_i)$ in a representation of a union $\cup_{i=1}^{k}\rho_i$ where $k\in\mathbb{N}$. For unions in which all modules have the same period, we write $\rho(A,h)=\cup\{\rho(a,h)~|~a\in{A}\}$ where $h\in\mathbb{N}$ and $A\subseteq[0,h-1]$ is the set of residues. We now present some preliminary definitions and results on unions of modules.

\begin{lemma}\label{lem:normal-form-lemma}
Let $b\mid{h}$, then $\rho(a,b)=\rho(A,h)$ where $$A=\left\{a+bj~|~j\in\left[0,\frac{h}{b}-1\right]\right\}.$$
\end{lemma}
\begin{proof}
We first prove that $\rho(a,b)\subseteq\rho(A,h)$. If $x\in\rho(a,b)$ then $x=a+b\alpha$ for some $\alpha\in\mathbb{Z}$. If we divide $\alpha$ by $\frac{h}{b}$, we obtain $\alpha=\frac{h}{b}\alpha_b+j$ where $j\in\left[0,\frac{h}{b}-1\right]$ and $\alpha_b\in\mathbb{Z}$. Thus, $x=a+bj+h\alpha_b$, so $x\in\rho(A,h)$. Now, we prove that $\rho(a,b)\supseteq\rho(A,h)$. If $x\in\rho(A,h)$, then $x=a+bj+h\alpha$ where $j\in\left[0,\frac{h}{b}-1\right]$. Since $\frac{h}{b}\in\mathbb{Z}$, we conclude that $x=a+b\left(j+\alpha\frac{h}{b}\right)\in\rho(a,b)$.
\end{proof}

We say that a union of modules is in \emph{normal form} if all modules have the same period, that is, if it is expressed as $\rho(A,h)$. Any union of modules can be transformed into normal form. To this end, we set $h$ equal to the least common multiple of the periods of the modules in the union. For each module in the union, Lemma~\ref{lem:normal-form-lemma} produces a set of modules, and the union of all these sets constitutes a normal form.

\begin{proposition}\label{pro:normal-form}
Every union of modules has a representation in normal form.
\end{proposition}

We denote by $A_h(\rho)$ the set $\left\{x\bmod{h}~|~x\in\rho\right\}$ where $x\bmod{y}$ is the remainder of $x$ upon division by $y$. The first part of the following result is known and stated here in a form adapted to our setting.
\begin{proposition}\label{pro:periods-are-divisors-of-b}
\sloppy{The set $A_h(\rho(a,b))$ is equal to $$\left\{b^*\alpha+a^*~|~\alpha\in\left[0,\frac{h}{b^*}-1\right]\right\}$$ where $b^*=\gcd(h,b)$ and $a^*=(a\mod{b^*})$. Moreover, $\rho(a,b)\subseteq\rho(A_h(\rho(a,b)),h)=\rho(a^*,b^*)$.}
\end{proposition}
\begin{proof}
By the Bézout's identity and multiplying by  $\alpha^*\in\mathbb{Z}$, we have that there exist $q,r\in\mathbb{Z}$ such that $br\alpha^*=b^*\alpha^*-hq\alpha^*$. Since $a=b^*\alpha_{a}+a^*$ for a particular $\alpha_{a}\in\mathbb{Z}$, then $$a+br\alpha^*=b^*(\alpha_{a}+\alpha^*)+a^*-hq\alpha^*.$$
It occurs that for all $\alpha\in\left[0,\frac{h}{b^*}-1\right]$, we can find $\alpha^*\in\mathbb{Z}$ such that $\alpha_a+\alpha^*=\alpha$. Thus, for all $\alpha\in\left[0,\frac{h}{b^*}-1\right]$ there is $x\in\rho(a,b)$ such that $x=b^*\alpha+a^*+h\alpha_h$ for some $\alpha_h\in\mathbb{Z}$. Moreover, we can observe that $0\leq{b^*\alpha+a^*}<h$ for all $\alpha\in\left[0,\frac{h}{b^*}-1\right]$. Therefore, $$\left\{b^*\alpha+a^*~|~\alpha\in\left[0,\frac{h}{b^*}-1\right]\right\}\subseteq{A_h(\rho)}.$$

By Proposition~\ref{pro:congruence-class-contained-condition} it occurs that $\rho\subseteq\rho^*$ because $b^*\mid{b}$ and $b^*\mid(a-a^*)$. Moreover, by Lemma~\ref{lem:normal-form-lemma} we have that $\rho(a^*,b^*)=\rho\left(\left\{b^*\alpha+a^*~|~\alpha\in\left[0,\frac{h}{b^*}-1\right]\right\},h\right).$ Thus, $x\in\rho(a,b)$ is equal to $a^*+b^*\alpha+h\alpha_h$ where $\alpha\in\left[0,\frac{h}{b^*}-1\right]$ and $\alpha_h\in\mathbb{Z}$. Therefore, $$A_h(\rho)\subseteq\left\{b^*\alpha+a^*~|~\alpha\in\left[0,\frac{h}{b^*}-1\right]\right\}.$$
\end{proof}

We can observe that if $\rho(a,b)$ is contained in a union of modules~$U$ with normal form $\rho(A,h)$, then by Proposition~\ref{pro:periods-are-divisors-of-b} we have that $A_h(\rho)$ must be contained in $A$. Moreover, since $A_h(\rho)$ is contained in $A$, it follows that $\rho^*=\rho(A_h(\rho),h)$ is contained in $U$. Thus, we obtain the following result from the previous proposition.

\begin{corollary}\label{cor:periods-are-divisors-of-b}
Let $U=\rho(A,h)$ and $\rho\subseteq{U}$, then
$A_h(\rho)\subseteq{A}$. Moreover, $\rho\subseteq\rho^*\subseteq{U}$ where $b^*=\gcd(h,b)$ and $a^*=(a\mod{b^*})$.
\end{corollary}

Given a union of modules $U$, a module $\sigma\subseteq{U}$ is \emph{maximal} in~$U$ if there is no module $\rho\subseteq{U}$ such that $\sigma\varsubsetneq\rho$. The set of all maximal modules in $U$ is denoted by $C(U)$. We use the symbol $\sigma$ instead of $\rho$ to distinguish maximal modules from those that are not necessarily maximal. 
\begin{proposition}\label{pro:maximal-congruence-classes-cover-residues}
Let $U=\rho(A,h)$ and $\sigma(c,d)\in{C(U)}$, then $d\mid{h}$.
\end{proposition}
\begin{proof}
Since $\sigma(c,d)\subseteq{U}$, by Corollary~\ref{cor:periods-are-divisors-of-b} there is $\sigma'=\sigma(c',d')$ such that $\sigma\subseteq\sigma'\subseteq{U}$ where $d'=\gcd(h,d)$ and $c'=(c\mod{d'})$. If $\sigma\neq\sigma'$, we get a contradiction because $\sigma$ is maximal. Therefore, $\sigma=\sigma'$, so $d=d'$ and $d\mid{h}$ because $d'\mid{h}$ .
\end{proof}

Every module contained in a union of modules is contained in some maximal module.
\begin{proposition}\label{pro:module-contained-in-maximal}
$\rho\subseteq{U}$ if and only if there exists $\sigma\in{C(U)}$ such that $\rho\subseteq\sigma$.
\end{proposition}
\begin{proof}
To begin, we show that if $\rho(a,b)\subseteq{U}$, then there exists $\sigma(c,d)\in{C(U)}$ such that $\rho\subseteq\sigma$. We can take the representative of~$\rho$ to be the representative of $\sigma$. Given a module $\rho$, we consider all the modules in which it is contained. By Proposition~\ref{pro:congruence-class-contained-condition}, we know that these modules are $\sigma$ where $d\mid{b}$. Note that if $\sigma$ is not contained in $U$, for all $d<{b}$ we have that $\rho$ belongs to $C(U)$ by the definition of maximal module. Otherwise, we can take the smallest divisor $d$ of $b$ for which $\sigma\subseteq{U}$ and for any other divisor $d'$ of $d$ different from $d$, we have that $\sigma'\not\subseteq{U}$. Thus, $\sigma$ is a maximal module in $U$ and $\rho\subseteq\sigma$. Observe that we can take $d=1$ as the smallest divisor of $b$ satisfying this condition if none of the others do. The converse is clear because $\rho\subseteq\sigma$ and $\sigma\subseteq{U}$ by definition.
\end{proof}

A union of modules coincides with the union of its maximal modules, and in particular, the smallest representation consists of maximal modules.
\begin{proposition}\label{pro:smallest-representation-maximal-modules}
The smallest representation of $U$ consists of maximal modules.
\end{proposition}
\begin{proof}
We have that $U\subseteq\cup\{\sigma~|~\sigma\in{C(U)}\}$ because, by Proposition~\ref{pro:module-contained-in-maximal}, any module $\rho_i$ that forms part of the representation of $U$ is contained in at least one maximal module. Finally, by the definition of maximal module of $U$, it is clear that $U\supseteq\cup\{\sigma~|~\sigma\in{C(U)}\}$. Thus, $U=\cup\{\sigma~|~\sigma\in{C(U)}\}$, so the union of maximal modules is a representation of $U$. Now, suppose that we found a minimal-size representation that has a module $\rho$ that is not maximal. By Proposition~\ref{pro:module-contained-in-maximal}, we have that it is contained in a maximal module $\sigma$. Since $\rho$ is not maximal, we have that $\rho\subsetneq\sigma$. It follows from Proposition~\ref{pro:congruence-class-contained-condition} that $d\mid{b}$, and $d<b$ since $\sigma\neq\rho$. The residue $c$ of $\sigma$ is less than or equal to the residue $a$ of $\rho$. Therefore, the size of $\sigma$ is strictly less than the size of $\rho$, and replacing $\rho$ with $\sigma$ the result is a representation of $U$ with a smaller size. We obtain a contradiction to the assumption that there was a non-maximal module in the smallest representation, so it consists only of maximal modules.
\end{proof}

\subsection{Union of intersecting affine modules}

We say that a representation $\cup_{i=1}^k\rho_i$ of a union of modules is a \emph{union of intersecting modules} if $\cap_{i=1}^k\rho_i\neq\varnothing$. We use $U_*$ to denote a union of modules admitting such a representation and call $\cap_{i=1}^k\rho_i$ the \emph{set of common elements}. By the generalization of the Chinese remainder theorem to non-coprime moduli~\cite[Theorem~3.12]{jones1998elementary}, this set is itself an affine module. Every maximal module of a union of intersecting modules contains the set of common elements.
\begin{proposition}\label{pro:all-contain-q}
Let $U_*=\cup_{i=1}^k\rho_i$ such that $\cap_{i=1}^k\rho_i\neq\varnothing$. Then for all $\sigma\in{C(U_*)}$, it holds that $\cap_{i=1}^k\rho_i\subseteq\sigma$.
\end{proposition}
\begin{proof}
We take an arbitrary maximal module $\sigma(c,d)$ of $U_*$ and an arbitrary common element $q$ and show that $q\in\sigma$. By definition of maximality, we have $\sigma\subseteq{U_*}$. If $\sigma=\varnothing$, then $\sigma\subsetneq\rho_i$ for every $i=1,\ldots,k$, contradicting the maximality of $\sigma$. Thus, $\sigma\neq\varnothing$. Since $\sigma \subseteq U_*=\bigcup_{i=1}^k\rho_i$, it follows that there must exist at least one module $\rho_i$ such that $\sigma\cap\rho_i\neq\varnothing$.
Let $\{\rho_{i_1},\ldots,\rho_{i_s}\} \subseteq \{\rho_1,\ldots,\rho_k\}$ be the subfamily of modules that intersect $\sigma$. Thus, $\sigma$ is contained in the union of $\rho_{i_1},\ldots,\rho_{i_s}$, and it is also a maximal module in this union. By Proposition~\ref{pro:normal-form} the union of $\rho_{i_1},\ldots,\rho_{i_s}$ can be written in normal form, and we can use the procedure described to obtain a representation $\rho(A,h)$ where $h=\mathrm{lcm}(b_{i_1},\dots,b_{i_s})$. It holds that $d\mid{h}$ by Proposition~\ref{pro:maximal-congruence-classes-cover-residues}.

We choose an element $x_j\in\sigma\cap\rho_{i_j}$ for each $j=1,\cdots,s$. Since $q\in\rho_i$ for all $i=1,\ldots,k$, we have that $x_j=q+b_{i_j}\alpha_j$ for all $j=1\ldots,s$ and some $\alpha_j\in\mathbb{Z}$ by Lemma~\ref{lem:every-element-is-a-representative}. Without loss of generality, we take $c=x_{1}$ as the representative of $\sigma$. Now, the equalities $x_{j}-x_{1}=d\alpha'_{j}$ where $\alpha'_j\in\mathbb{Z}$ are taken into consideration. We observe these equalities $b_{i_1}\alpha_1=b_{i_j}\alpha_j-d\alpha'_j$ for all $j=2,\ldots,s$ and prove that $d$ divides $b_{i_1}\alpha_1$. It is clear that $b_{i_1}\alpha_1$ is divisible by $\gcd(b_{i_1},d)$ because it divides $b_{i_1}$. Moreover, the other side of the equalities ensures that $b_{i_1}\alpha_1$ is also divisible by $\gcd(b_{i_j},d)$ for every $j=2,\ldots,s$ because it divides $b_{i_j}$ and $d$. Thus, $b_{i_1}\alpha_1$ is divisible by $\text{lcm}(\gcd(b_{i_1},d),\ldots,\gcd(b_{i_s},d))=\gcd(\text{lcm}(b_{i_1},\cdots,b_{i_s}),d)$ which is equal $d$ because $d\mid\text{lcm}(b_{i_1},\cdots,b_{i_s})$. Therefore, there is $\alpha^{\prime}\in\mathbb{Z}$ such that $b_{i_1}\alpha_1=d\alpha^{\prime}$. Since $c=q+b_{i_1}\alpha_1$, we have that $q+b_{i_1}\alpha_1-d\alpha^{\prime}=q\in\sigma$.
\end{proof}

Every maximal module $\sigma$ of $U_*$ contains points that belong only to $\sigma$. Thus, no maximal module in a union of intersecting modules is redundant.

\begin{proposition}\label{pro:only-belong-to-a-single-maximal}
For every $\sigma\in{C(U_*)}$ there exists $x\in\sigma$ such that $x\not\in\sigma'$ for every $\sigma'\in{C(U_*)}$ with $\sigma'\neq\sigma$.
\end{proposition}
\begin{proof}
We have that every maximal module of $U_*$ contains the common elements by Proposition~\ref{pro:all-contain-q}. We take $q$ a common element, and we want to prove that given a maximal module $\sigma(q,d)$, the point $q+d$ only belongs to $\sigma$. If $q+d$ were contained in a different maximal module $\sigma(q,d')$, then $\sigma\subsetneq\sigma'$ by Corollary~\ref{cor:membership-eqiuvalent-to-subset}, which contradicts the fact that $\sigma$ is maximal. Therefore, $\sigma$ has a point that only belongs to $\sigma$ and none of the other maximal modules.
\end{proof}

The smallest representation of $U_*$ can be obtained trivially from the set of maximal modules in $U_*$.

\begin{proposition}\label{pro:smallest-representation-made-of-maximal-cc}
The smallest representation of $U_*$ is given by $C(U_*)$.
\end{proposition}
\begin{proof}
We have that the smallest representation of $U_*$ consists of maximal modules by Proposition~\ref{pro:smallest-representation-maximal-modules}, and by Proposition~\ref{pro:only-belong-to-a-single-maximal} every maximal module of $U_*$ contains points that do not belong to the other maximal modules. Hence, all maximal modules necessarily appear in the smallest representation.
\end{proof}

\section{Learning algorithm}\label{sec:learning-algorithm}

We present an algorithm to learn a union of modules using equivalence and subset queries. We begin with the case of a single module and prove the correctness of the corresponding learning procedure. This procedure can be derived from the online algorithm of \cite{helmbold1992learning} when restricting it to the one-dimensional case and accounting for the connection between online mistake bounds and query complexity in \cite{littlestone1988learning,angluin2004queries}. For completeness, we include an adaptation that will be used later.

\subsection{Learning a single affine module} \label{subsec:one}

Algorithm~\ref{alg:single-module} shows the process of learning a single module $\rho$. Initially, the hypothesis is $\rho^0=\varnothing$. The first counterexample returned by the equivalence query $e$ is used as the representative $a=x_1$, and the period is set to $b_1=0$, so $\rho^1=\rho(x_1,0)$. After the second counterexample, the first non-zero period $b_2$ is obtained and subsequently modified with the incoming counterexamples. Thus, $\rho^n=\rho(a,b_n)$ where $b_n=\gcd(0,x_2-x_1,\ldots,x_n-x_1)$ for every $n\geq{2}$. Each counterexample satisfies $x_{n}\in{\rho\setminus\rho^{n-1}}$. Here, we use $\gcd(0,x)=\left|x\right|$. When learning a union of modules, the modules in the hypothesis are treated analogously.

\begin{algorithm}[hbt!]
\caption{Learning a module $\rho$.}\label{alg:single-module}
$a=e(\varnothing)$\; $b=0$\; $\rho = \rho(a, 0)$\;
\While{$(x=e(\rho))$}{
    $b=\gcd(x-a,b)$\; $\rho = \rho(a, b)$\;
}
\Return{$\rho$\;}
\end{algorithm}

The hypothesis $\rho^n$ denotes the module obtained after receiving $n$ counterexamples, where the superscript $n$ indicates the $n$th module in a sequence generated as above. It is equal to the target module $\rho$ after at most polynomially many counterexamples in the size of the largest counterexample $x_{\ell}$.

\begin{proposition}\label{pro:convergence-of-module}
There exists $n\leq\log_2(2\left|x_{\ell}\right|)+2$ such that $\rho^n=\rho$. Moreover, $\rho^0\subset\ldots\subset\rho^{n-1}\subset\rho^{n}$.
\end{proposition}
\begin{proof}
Since $x_n\in\rho\setminus\rho^{n-1}$ and the first non-zero period is $b_2=x_2-x_1$ we have that $b\mid(x_n-a)$, but $b_{n-1}\nmid(x_n-a)$ by Lemma~\ref{lem:element-contained-in-cartesian-congruence-class} for $n\geq{3}$. Consequently, for $n\geq{3}$ the new period $b_{n}=\gcd(b_{n-1},x_n-a)$ is a multiple of $b$ which decreases at least by one prime factor at every iteration. Since $\frac{b_2}{b}$ has at most $\log_2(\frac{b_2}{b})$ prime factors, there exists $n\leq\log_2(\frac{b_2}{b})+2$ such that $b_n=b$, so $\rho^n=\rho$. This upper bound is tight. We consider an adversarial strategy in which the adversary always returns a counterexample that reduces the period by exactly one prime factor. This achieves the stated upper bound, and such a choice of counterexample is always possible.

Finally, we observe that $n\leq\log_2\left(b_2\right)+2\leq\log_2(2\left|x_{\ell}\right|)+2$, and by Proposition~\ref{pro:congruence-class-contained-condition} for all $n'>{n}$, occurs that $\rho_{n'}\subset\rho_{n}$ since $b_n\mid{b_{n'}}$ and $b_{n'}\neq{b_n}$.
\end{proof}

Proposition~\ref{pro:convergence-of-module} provides a polynomial upper bound on the number of counterexamples required to learn a module, which corresponds to the number of iterations of the while loop in Algorithm~\ref{alg:single-module}. The loop consists only of a $\gcd$ operation, which can be computed in polynomial time in the size of the largest counterexample.
\begin{corollary}\label{lem:single-arithmetic-progression}
$\mathcal{C}_M$ is efficiently exactly learnable with equivalence queries.
\end{corollary}

\subsection{Learning a union of affine modules}

\subsubsection{Algorithm description}

To learn a union of modules, Algorithm~\ref{alg:union-of-main} refines the hypothesis upon receiving each counterexample from the equivalence query $e$. The variable $H$ is used to construct the hypothesis of the algorithm. It consists of tuples $(a,b)$, each representing a module $\rho(a,b)$.

\SetKwProg{Fn}{Procedure}{}{}
\SetKwFunction{FProcUH}{RefineHypothesis}
\SetKwFunction{FProcMCM}{RefineModules}

\begin{algorithm}[hbt!]
\caption{Learning a union of modules $U$.}\label{alg:union-of-main}

$H=\left\{\right\}$\;

\While{$(x=e(\cup\{\rho(a, b)~|~(a,b)\in{H}\}))$}{
$H=\FProcUH{H,x}$\;
}
\Return{$H$\;}
\end{algorithm}

If no module in the hypothesis is refined by the current counterexample $x$, the tuple $(x,0)$ is added to $H$ in Algorithm~\ref{alg:update-hypothesis}.

\begin{algorithm}[hbt!]
\caption{Hypothesis refinement.}\label{alg:update-hypothesis}
\Fn{\FProcUH($H,x$)}{

$H'=\FProcMCM{H,x}$

\If{$H'={H}$}{\Return $H\cup\{(x,0)\}$}
\Return $H'$
}
\end{algorithm}

In Algorithm~\ref{alg:update-suitable-group}, for each module $\rho(a,b)$ in the hypothesis, we check whether $\rho(a,\gcd(b,x-a))$ is contained in $U$ and, in that case, replace $\rho(a,b)$ with $\rho(a,\gcd(b,x-a))$.

\begin{algorithm}[hbt!]
\caption{Module refinement.}\label{alg:update-suitable-group}
\Fn{\FProcMCM($H,x$)}{

$H'=H$

\For{$(a,b)\in{H}$}{
	\If{$s(\rho(a,\gcd\left(b,x-a\right)))$}{
	    $H'=\left(H'\setminus\{(a,b)\}\right)\cup\{(a,\gcd\left(b,x-a\right))\}$
	}
}
\Return $H'$
}
\end{algorithm}

\subsubsection{Algorithm correctness}

We introduce some notation for the correctness proof and present the algorithm recursively. We denote by $\rho^r$ a module in the union of modules of the hypothesis, constructed as in Section~\ref{subsec:one}. We say that $\rho^{r}$ is \emph{refined} by a counterexample $x$ if the algorithm replaces $\rho^{r}=\rho(a,b_{r})$ with $\rho(a,b_{r+1})$ where $b_{r+1}=\gcd(b_{r},x-a)$. In this case, we write $\rho^{r+1}=\rho(a,b_{r+1})$. This occurs in the procedure $\FProcMCM$ whenever $\rho(a,b_{r+1})$ is contained in $U$.

Initially, the hypothesis is $H^0=\varnothing$ and $k_0=0$. Let
\[
H^n=\cup_{i=1}^{k_n}\rho_i^{r_n(i)}
\]
where $n$ is the number of counterexamples received and $k_n$ is the number of modules in the hypothesis after processing the $n$th counterexample. The exponent $r_n(i)$ denotes the number of refinements of the $i$th module after processing $n$ counterexamples.

Given a counterexample $x_{n+1}\in{U\setminus{H^n}}$ there are two cases for constructing $H^{n+1}$.

\textbf{Case 1:} No module in $H^n$ is refined by $x_{n+1}$. Then
\[
H^{n+1}=H^n\cup \rho(x_{n+1},0),
\]
so $k_{n+1}=k_n+1$ and the new module is initialized with $r_{n+1}(k_{n+1})=1$, as in $\FProcUH$.

\textbf{Case 2:} At least one module in $H^n$ is refined by $x_{n+1}$ in $\FProcMCM$. Then
\[
H^{n+1}=\cup_{i=1}^{k_n}\rho_i^{r_{n+1}(i)},
\]
where
\[
\rho_i^{r_{n+1}(i)}=\rho(a_i,b_{i,r_{n+1}(i)}),
\]
with $b_{i,r_{n+1}(i)}=\gcd(b_{i,r_{n}(i)},x_{n+1}-a_i)$
and $r_{n+1}(i)=r_n(i)+1$ if the $i$th module is refined, and $\rho_i^{r_{n+1}(i)}=\rho_i^{r_{n}(i)}$ with  $r_{n+1}(i)=r_n(i)$ otherwise. In this case, $k_{n+1}=k_n$.

At each iteration
of Algorithm~\ref{alg:union-of-main}, the hypothesis is contained in the target concept $U$.
\begin{proposition}\label{pro:hypothesis-contained}
$H^n\subseteq{U}$ for all $n\in\mathbb{N}$.
\end{proposition}
\begin{proof}
We prove this by induction. For the base case, we have that $H^0=\varnothing$ is contained in $U$. Suppose that $H^n\subseteq{U}$, we prove that $H^{n+1}\subseteq{U}$ after processing the counterexample $x_{n+1}$. There are two cases. In the first case, none of the modules in $H^n$ is refined by $x_{n+1}$ and $H^{n+1}=H^n\cup\{\rho(x_{n+1},0)\}$. Since $H^n\subseteq{U}$, then $x_{n+1}\in{U}$ because the equivalence query returns an element of $(U\setminus{H^n})\cup(H^n\setminus{U})=(U\setminus{H^n})$. Therefore, $H^{n+1}\subseteq{U}$ because $\rho(x_{n+1},0)=\{x_{n+1}\}\subseteq{U}$. In the second case, $H^{n+1}=\cup_{i=1}^{k_n}\rho_i^{r_{n+1}(i)}$ where by construction of the algorithm $\rho_i^{r_{n+1}(i)}\subseteq{U}$ for all $i=1,\ldots,k_n$. Hence, $H^{n+1}\subseteq{U}$.
\end{proof}

We overload the notation $C(\cdot)$. Given a single module~$\rho$, $C(\rho)$ denotes the set of all maximal modules $\sigma\in{C(U)}$ such that $\rho\subseteq\sigma$. We call $C(\rho)$ the \emph{maximal covering} of $\rho$ in $U$, omitting $U$ when it is clear from the context. If a module is contained in another module, then the maximal covering of the latter is contained in the maximal module of the former.
\begin{lemma}\label{lem:set-of-modules-to-which-it-can-converge-decrease}
If $\rho\subseteq\rho'$, then $C(\rho')\subseteq{C(\rho)}$.
\end{lemma}
\begin{proof}
Let $\sigma\in{C(\rho')}$, by the definition of maximal covering $\rho'\subseteq\sigma$. Since $\rho\subseteq\rho'\subseteq\sigma$, we have that $\sigma\in{C(\rho)}$.
\end{proof}

A module in the hypothesis is refined by a counterexample if and only if the counterexample belongs to a module in its maximal covering.
\begin{proposition}\label{pro:condition-to-be-contained}
Given $x\in{U}$, we have $\rho(a,\gcd(b,x-a))\subseteq{U}$ if and only if there exists $\sigma\in{C(\rho)}$ such that $x\in\sigma$.
\end{proposition}
\begin{proof}
Since $\rho(a,b)$ is refined by $x$, we have that $\rho'=\rho(a,b')$ where $b'=\gcd(b,x-a)$ is contained in $U$. Then, we have by Proposition~\ref{pro:module-contained-in-maximal} that $C(\rho')\neq\varnothing$. It follows from Lemma~\ref{lem:set-of-modules-to-which-it-can-converge-decrease} that $C(\rho')\subseteq{C(\rho)}$ because $\rho\subseteq\rho'$ by Proposition~\ref{pro:congruence-class-contained-condition}. Thus, there is $\sigma\in{C(\rho)}$ such that $\rho'\subseteq\sigma$. Therefore, $x\in\sigma$ since $x\in\rho'$ by Lemma~\ref{lem:element-contained-in-cartesian-congruence-class} because $b'\mid(x-a)$.

It remains to prove the opposite implication. Since $x\in\sigma(c,d)$ we have that $d\mid(x-c)$ by Lemma~\ref{lem:element-contained-in-cartesian-congruence-class}. Moreover, by the definition of maximal covering $\rho\subseteq\sigma$, so $d\mid{b}$ and $d\mid(a-c)$ by Proposition~\ref{pro:congruence-class-contained-condition}. We can conclude that $d\mid(x-a)$ because $d$ divides $x-c-(a-c)$. Therefore, $d\mid{b'}$ since any common divisor of $b$ and $x-a$ divides $\gcd(b,x-a)$. Taking $d\mid{b'}$ and $d\mid(a-c)$ together, we obtain by Proposition~\ref{pro:congruence-class-contained-condition} that $\rho'\subseteq\sigma\subseteq{U}$. Therefore, $\rho$ is refined by $x$.
\end{proof}

Only a polynomial number of counterexamples is required for a module in the hypothesis to become maximal.
\begin{proposition}\label{pro:convergence-of-modules}
Let $\rho^r$ be a module in a hypothesis. There exists $r\leq\log_2(2\left|x_{\ell}\right|)+2$ such that $\rho^{r}$ is equal to a maximal module $\sigma$. For such an $r$, we have $C(\rho^{r})=\{\sigma\}$, and the equivalence query no longer returns counterexamples that refine $\rho^{r}$.
\end{proposition}
\begin{proof}
A module $\rho^r$ in the hypothesis is contained in $U$ for every $r\in\mathbb{N}$ by Proposition~\ref{pro:hypothesis-contained}. Thus, for all $r\in\mathbb{N}$ we have that $\rho^r$ is contained in at least one maximal module $\sigma$ by Proposition~\ref{pro:module-contained-in-maximal}. We also have that the counterexamples by which it was refined $x_{i_1},\ldots,x_{i_r}$ belong to $\rho^r$ for all $r\in\mathbb{N}$ because $a=x_{i_1}$ and $b_r=\gcd(x_{i_2}-a,\ldots,x_{i_r}-a)$ divides $x_{i_j}-a$ for all $j=2,\ldots,r$. By Proposition~\ref{pro:convergence-of-module}, we conclude that there exists $r\leq\log_2(2\left|x_{\ell}\right|)+2$ such that $\rho^r=\sigma$.

By definition, a maximal module is not contained in any other maximal module. Hence, $C\left(\rho^{r}\right)=\{\sigma\}$ since for this particular $r$, we have $\rho^{r}=\sigma$. Moreover, $\sigma \subseteq H$ and $\sigma \subseteq U$ and therefore $\sigma\cap((U\setminus{H})\cup(H\setminus{U}))=\varnothing$. By Proposition~\ref{pro:condition-to-be-contained}, the only counterexamples that refine $\rho^{r}$ are those from $\sigma$, whereas the equivalence query returns only examples from $((U\setminus{H})\cup(H\setminus{U}))$. Consequently, once this $r$ is reached, the equivalence query no longer returns counterexamples that refine $\rho^{r}$.
\end{proof}

The maximal coverings of the modules in the hypothesis are pairwise disjoint at each iteration of Algorithm~\ref{alg:union-of-main}. This will be an important argument in showing that the number of modules in the hypothesis is at most the number of maximal modules.
\begin{proposition}\label{pro:disjointness-of-coverings}
$\{C(\rho)~|~\rho\in{H^n}\}$ is pairwise disjoint for all $n\in\mathbb{N}$.
\end{proposition}
\begin{proof}
By induction, the base case $H^0=\varnothing$  is immediate. Assume $H^n$ satisfies the condition. For the counterexample $x_{n+1}$, there are two cases. In the first case, none of the modules $\rho\in{H^n}$ is refined by $x_{n+1}$, and $H^{n+1}=H^n\cup\{\rho(x_{n+1},0)\}$. It follows from Proposition~\ref{pro:condition-to-be-contained} that for all $\rho\in{H^n}$ and all $\sigma\in{C(\rho)}$ we have $x_{n+1}\not\in\sigma$. Moreover, for all $\sigma'\in{C(\rho(x_{n+1},0))}$ we have that $x_{n+1}\in\sigma'$ because $x_{n+1}\in\rho(x_{n+1},0)$ and $\rho(x_{n+1},0)\subseteq\sigma'$ by definition of maximal covering. Thus, $C(\rho(x_{n+1},0))\cap{C(\rho)}=\varnothing$ for every $\rho\in{H^n}$, so $\{C(\rho)~|~\rho\in{H^{n+1}}\}$ is pairwise disjoint. In the second case, $H^{n+1}=\{\rho_1^{r_{n+1}(1)},\ldots,\rho_{k_n}^{r_{n+1}(k_n)}\}$ and $\rho_i^{r_{n}(i)}\subseteq\rho_i^{r_{n+1}(i)}$ for all $i=1,\ldots,k_n$. Hence, by Lemma~\ref{lem:set-of-modules-to-which-it-can-converge-decrease} for all $i=1,\ldots,k_n$ we have $C(\rho_i^{r_{n+1}(i)})\subseteq{C(\rho_i^{r_{n}(i)})}$. Therefore, $\{C(\rho)~|~\rho\in{H^{n+1}}\}$ is pairwise disjoint after processing $x_{n+1}$.
\end{proof}

We overload the notation $C(\cdot)$. For a hypothesis $H$, we define $C(H)$ as $\cup\{C(\rho)~|~\rho\in{H}\}$. At each iteration of Algorithm~\ref{alg:union-of-main}, the number of modules in $H$ is at most $|C(H)|$.
\begin{proposition}\label{pro:at-most-h}
$k_n\leq\left|C(H^n)\right|$ for all $n\in\mathbb{N}$.
\end{proposition}
\begin{proof}
For every $n\in\mathbb{N}$, we have
$$k_n\leq\sum_{i=1}^{k_n}|C(\rho_i^{r_n(i)})|=|C(H^n)|.$$
The inequality is given by Proposition~\ref{pro:hypothesis-contained} and Proposition~\ref{pro:module-contained-in-maximal}, since $H^n\subseteq{U}$ then $\rho_i^{r_n(i)}\subseteq{U}$ for all $i=1,\ldots,k_n$. The equality is given by Proposition~\ref{pro:disjointness-of-coverings}.
\end{proof}

At each iteration
of Algorithm~\ref{alg:union-of-main}, the number of modules in the hypothesis is at most the number of maximal modules, since $C(H)\subseteq{C(U)}$.
\begin{corollary}\label{cor:no-more-than-k}
$k_n\leq\left|C(U)\right|$ for all $n\in\mathbb{N}$.
\end{corollary}

The number of counterexamples required by the algorithm to learn a union of modules is polynomial in the size of the smallest representation and in the size of the largest counterexample.
\begin{theorem}\label{pro:number-of-counterexamples}
There exists $n\leq{|C(U)|}\log_2(2\left|x_{\ell}\right|)+2|C(U)|$ such that $H^n=U$.
\end{theorem}
\begin{proof}
If the equivalence query returns no counterexample, then, by definition of the query, the current hypothesis $H$ coincides with the target $U$. Therefore, the algorithm terminates and returns a correct hypothesis. We prove that the number of counterexamples returned is at most $|C(U)|\log_2(2\left|x_{\ell}\right|)+2|C(U)|$.

By Corollary~\ref{cor:no-more-than-k}, for the purposes of this proof we may assume that we start with $|C(U)|$ modules in the hypothesis, none of which has been refined by any counterexample. Thus, $H^0=\cup_{i=1}^{|C(U)|}\rho_i^{r_0(i)}$ where $r_0(i)=0$ and $\rho_i^0=\varnothing$ for all $i=1,\ldots,|C(U)|$. For modules $\rho_i^{r_n(i)}$ that have not been refined (i.e., $r_n(i)=0$), we use the convention $C(\rho_i^{r_n(i)})=C(U)\setminus{C(H^n )}$. This convention ensures that for all $\sigma\in{C(U)}$ and all $n\in\mathbb{N}$, there exists $\rho$ in $H^n$ such that $\sigma\in{C(\rho)}$. We now explain this. By Proposition~\ref{pro:at-most-h}, if $C(H^n)$ does not contain all maximal modules, then the number of refined modules with $r_n(i)\geq{0}$ is less than $|C(U)|$. Hence, there exists at least one module such that $r_n(i)=0$, whose maximal covering is $C(U)\setminus{C(H^n)}$ by the convention above.

Given a counterexample $x_{n+1}$, by Proposition~\ref{pro:smallest-representation-maximal-modules} there exists a maximal module $\sigma$ such that $x_{n+1}\in\sigma$. As we saw in the previous paragraph, for all $\sigma\in{C(U)}$ there exists $\rho_i^{r_n(i)}$ in $H^n$ such that $\sigma\in{C(\rho_i^{r_n(i)})}$. By Proposition~\ref{pro:condition-to-be-contained} we have that $\rho_i^{r_n(i)}$ is refined by $x_{n+1}$, so $r_{n+1}(i)=r_n(i)+1$. Since for every counterexample at least one module is refined, we have that $n\leq\sum_{i=1}^{|C(U)|}r_n(i)$.

For each $i=1,\ldots,|C(U)|$ there exists $r_i\leq\log_2(2\left|x_{\ell}\right|)+2$ such that $\rho_i^{r_i}$ is equal to a maximal module and no more counterexamples that refine $\rho_{i}^{r_{n}(i)}$ are returned by the equivalence query once $r_{n}(i)=r_{i}$, as stated in Proposition~\ref{pro:convergence-of-modules}. Thus, we have $r_{n}(i)\leq{r_{i}}$ for all $n\in\mathbb{N}$. Since $r_n(i)$ is initially equal $0$, for every $n$ each refined module increases $r_n(i)$ by one and at least one is refined, there exists $n$ for which $r_n(i)=r_i$ for all $i=1,\ldots,|C(U)|$. For that particular $n$, we have that  $n\leq\sum_{i=1}^{|C(U)|}r_n(i)=\sum_{i=1}^{|C(U)|}r_i$. It follows that there exists $n\leq{|C(U)|}\log_2(2\left|x_{\ell}\right|)+2|C(U)|$ such that every module in $H^n$ is equal to a maximal module. By Proposition~\ref{pro:convergence-of-modules} and Proposition~\ref{pro:disjointness-of-coverings}, the hypothesis has $|C(U)|$ different maximal modules. Hence, $H^n$ is equal to $U$ by Proposition~\ref{pro:smallest-representation-maximal-modules}.
\end{proof}

\section{Efficiency of learning a union of intersecting affine modules}\label{sec:intersecting}

We show that $\mathcal{C}^*_U$ is efficiently exactly learnable with equivalence and subset queries and that subset queries can be replaced by membership queries when a common element of the modules in the union is known. Building on the results above, we now state the following theorem.

\begin{theorem}\label{the:efficiency-intersecting}
$\mathcal{C}^*_U$ is efficiently exactly learnable with equivalence and subset queries.
\end{theorem}
\begin{proof}
Algorithm~\ref{alg:union-of-main} computes a hypothesis that is equal to $U_*$ after at most $|C(U_*)|\log_2(2\left|x_{\ell}\right|)+2|C(U_*)|$ counterexamples, as it is shown in Theorem~\ref{pro:number-of-counterexamples}. \sloppy{Thus, the while loop in Algorithm~\ref{alg:union-of-main} is executed at most $|C(U_*)|\log_2(2\left|x_{\ell}\right|)+2|C(U_*)|$ times. Moreover, the loop inside $\FProcMCM$ is executed at most $|C(U_*)|$ times by Corollary~\ref{cor:no-more-than-k}, and the execution time of $\gcd$ is polynomial in the size of the largest counterexample. We have by Proposition~\ref{pro:smallest-representation-made-of-maximal-cc} that $|C(U_*)|$ is equal to the number of modules in the smallest representation. Since $|C(U_*)|\leq\text{size}(U_*)$, we have that the algorithm is polynomial in the size of the smallest representation and the size of the largest counterexample.}
\end{proof}

\subsection{Learning with a known common element}

We begin by describing a modified algorithm for the case in which a common element $q$ of the union of intersecting modules is known. In Algorithm~\ref{alg:union-of-main-2}, the hypothesis is refined upon receiving each counterexample from an equivalence query.

\begin{algorithm}[hbt!]
\caption{Learning $U_*$ using a common element $q$.}\label{alg:union-of-main-2}

$H=\left\{\right\}$\;

\While{$(x=e(\cup\{\rho(q, b)~|~(q, b)\in{H}\}))$}{
$H=\FProcUH{H,x}$\;
}
\Return{$H$\;}
\end{algorithm}

If none of the modules in the hypothesis is refined by the current counterexample $x$, the algorithm creates a new module with the common element $q$ as the representative and period equal $x-q$, as in Algorithm~\ref{alg:update-hypothesis-2}.

\begin{algorithm}[hbt!]
\caption{Hypothesis refinement.}\label{alg:update-hypothesis-2}
\Fn{\FProcUH($H,x$)}{

$H'=\FProcMCM{H,x}$

\If{$H'={H}$}{\Return $H\cup\{(q,x-q)\}$}
\Return $H'$
}
\end{algorithm}

We can use membership queries rather than subset queries in the learning algorithm, as a consequence of the following result.
\begin{proposition}\label{pro:belong-to-a-union}
$q+b\in{U_*}$ if and only if $\rho(q,b)\subseteq{U_*}$ where $q$ is a common element.
\end{proposition}
\begin{proof} Assume that $q+b\in{U_*}$ then there exists a maximal module $\sigma$ such that $q+b\in\sigma$. By Proposition~\ref{pro:all-contain-q} we have that $\sigma$ contains $q$, so $\sigma=\sigma(q,d)$ and by Corollary~\ref{cor:membership-eqiuvalent-to-subset} we have that $\rho(q,b)\subseteq\sigma(q,d)\subseteq{U_*}$. The converse is clear since $q+b\in\rho(q,b)$ by Lemma~\ref{lem:element-contained-in-cartesian-congruence-class}.
\end{proof}

In Algorithm~\ref{alg:update-suitable-group}, for each module $\rho(q,b)$ in the hypothesis, we check whether $q+\gcd(b,x-q)$ belongs to $U$ and in that case replace $\rho(q,b)$ with $\rho(q,\gcd(b,x-q))$. By Proposition~\ref{pro:belong-to-a-union}, this occurs whenever $\rho(q,\gcd\left(b,x-q\right))$ is contained in the union.

\begin{algorithm}[hbt!]
\caption{Module refinement.}\label{alg:update-suitable-group-2}
\Fn{\FProcMCM($H,x$)}{

$H'=H$

\For{$(q,b)\in H$}{
	\If{$m(q+\gcd\left(b,x-q\right))$}{
	    $H'=\left(H'\setminus\{(q,b)\}\right)\cup\{(q,\gcd\left(b,x-q\right))\}$
	}
}
\Return $H'$
}
\end{algorithm}

The following propositions are used to show that if Algorithm~\ref{alg:union-of-main} is correct and efficient, then so is Algorithm~\ref{alg:union-of-main-2}. From the beginning, the corresponding modules in both hypotheses are refined by the same counterexamples.

\begin{proposition}\label{pro:new-modules-same-covering}
Let $x\in{U^*}$ and a common element $q$. The maximal coverings of $\rho(x,0)$ and $\rho(q,x-q)$ are equal.
\end{proposition}
\begin{proof}
First we prove that $C(\rho(q,x-q))\subseteq{C(\rho(x,0))}$. Since $\{x\}=\rho(x,0)$ we have that $\rho(x,0)\subseteq\rho(q,x-q)$ by Lemma~\ref{lem:element-contained-in-cartesian-congruence-class}. Hence, the inclusion follows by  Lemma~\ref{lem:set-of-modules-to-which-it-can-converge-decrease}. Now, we prove that $C(\rho(x,0))\subseteq{C(\rho(q,x-q))}$. We take an arbitrary maximal module $\sigma$ from $C(\rho(x,0))$. It contains $\{x\}=\rho(x,0)$ by the definition of maximal covering, and we have that $\sigma$ also contains $q$ by Proposition~\ref{pro:all-contain-q}. Therefore, $\rho(q,x-q)\subseteq\sigma$ by Proposition~\ref{pro:module-with-x-q-is-contained}.
\end{proof}

After each refinement, the corresponding modules continue to be refined by the same counterexamples.
\begin{proposition}\label{pro:same-maximal-covering}
Let $\rho^r$ and $\overline{\rho}^r$ be modules in a hypothesis with $C(\rho^r)=C(\overline{\rho}^r)$. Let $\rho^{r+1}$ and $\overline{\rho}^{r+1}$ denote their refinements by the same counterexample. Then $C(\rho^{r+1})=C(\overline{\rho}^{r+1})$.
\end{proposition}
\begin{proof}
We can observe that $C(\rho^{r+1})$ does not contain any maximal module from $C(\rho^r)$ which does not contain $x$ because $\rho^{r+1}=\rho(a,\gcd(b_r,x-a))$ contains $x$. On the other hand, suppose that $x$ belongs to a maximal module $\sigma(c,d)$ from $C(\rho^r)$. Since $\rho^r\subseteq\sigma$ we have that $d\mid{b_r}$ and $d\mid(x-a)$, so $d\mid\gcd(b_r,x-a)$. Therefore, $\rho^{r+1}\subseteq\sigma$ and $\sigma$ belongs to $C(\rho^{r+1})$. Thus, $C(\rho^{r+1})=\{\sigma\in{C(\rho^r)}~|~x\in\sigma\}$. Analogously, since $C(\rho^r)=C(\overline{\rho}^r)$ we have that $C(\overline{\rho}^{r+1})=\{\sigma\in{C(\rho^r)}~|~x\in\sigma\}$.
\end{proof}

We now state the main result of this section.
\begin{theorem}
When a common element is known, $\mathcal{C}^*_U$ is efficiently exactly learnable with equivalence and membership queries.
\end{theorem}
\begin{proof}
We simulate the execution of Algorithm~\ref{alg:union-of-main} and Algorithm~\ref{alg:union-of-main-2} in parallel on the same $U_*$ and observe that, if the former learns $U_*$, then so does the latter. Concretely, from our reasoning we can deduce that the number of counterexamples required by Algorithm~\ref{alg:union-of-main-2} is less than or equal to the number required by Algorithm~\ref{alg:union-of-main}. Moreover, the hypotheses produced by both algorithms have the same number of modules. With this in mind, a similar argument to that used in the proof of Theorem~\ref{the:efficiency-intersecting} can be applied to obtain the result. We prove it by induction in the number of counterexamples $n$. We write $H^n$ for the hypothesis of Algorithm~\ref{alg:union-of-main} and $\overline{H^{n}}$ for the hypothesis of Algorithm~\ref{alg:union-of-main-2}. The induction hypothesis is that, for every $\rho^{r_n(i)}_i$ in $H^n$ there is a corresponding module $\overline{\rho}^{r'_n(i)}_i$ in $\overline{H^{n}}$ such that $\rho^{r_n(i)}_i\subseteq\overline{\rho}^{r'_n(i)}_i$ and $C(\rho^{r_n(i)}_i)=C(\overline{\rho}^{r'_n(i)}_i)$ and every module in $\overline{H}^n$ has a corresponding module in $H^n$.

In the base case, both $H^0$ and $\overline{H^{0}}$ contain zero modules, so the induction hypothesis is met. Assume that the induction hypothesis holds for $n$. We show that it also holds for $n+1$. There are two cases for a counterexample $x_{n+1}$. In the first case none of the modules in $H^n$ is refined. Since every module in $\overline{H^{n}}$ has the same maximal covering as its corresponding module in $H^n$, if none of the modules in $H^n$ was refined by $x_{n+1}$ then none of the modules in $\overline{H^{n}}$ is refined by $x_{n+1}$ by Proposition~\ref{pro:condition-to-be-contained} and Proposition~\ref{pro:belong-to-a-union} because $C(\rho^{r_n(i)}_i)=C(\overline{\rho}^{r'_n(i)}_i)$ for all $i=1,\ldots,k_n$. Thus, $H^{n+1}=H^n\cup\rho(x,0)$, and $\overline{H^{n+1}}=\overline{H^{n}}\cup\rho(q,x-q)$. We have that $\rho(x,0)\subseteq\rho(q,x-q)$ by Lemma~\ref{lem:element-contained-in-cartesian-congruence-class}, and their maximal coverings are equal by Proposition~\ref{pro:new-modules-same-covering}. In the other case, $H^{n+1}=\cup_{i=1}^{k_n}\rho_i^{r_{n+1}(i)}$ because some module was refined. Since $C(\rho^{r_n(i)}_i)=C(\overline{\rho}^{r'_n(i)}_i)$ for all $i=1,\ldots,k_n$, if some of $\rho^{r_n(i)}_i$ is refined, then its corresponding module $\overline{\rho}^{r'_n(i)}_i$ is refined by Proposition~\ref{pro:condition-to-be-contained} and Proposition~\ref{pro:belong-to-a-union}. Thus, $\overline{H^{n+1}}=\cup_{i=1}^{k_n}\overline{\rho}_i^{r'_{n+1}(i)}$. Since $\rho^{r_n(i)}_i\subseteq\overline{\rho}_i^{r'_{n}(i)}$ we have that $\overline{b}_{i,r'_n(i)}\mid{b_{i,r_n(i)}}$, so $\gcd(\overline{b}_{i,r'_n(i)},x-a)\mid\gcd\left(b_{i,r_n(i)},x-a\right)$ where $a\in\rho^{r_n(i)}_i$. Therefore, $\rho^{r_{n+1}(i)}_i\subseteq\overline{\rho}_i^{r'_{n+1}(i)}$ by Proposition~\ref{pro:congruence-class-contained-condition}. Finally, we have that $C(\rho^{r_{n+1}(i)}_i)=C(\overline{\rho}_i^{r'_{n+1}(i)})$ by Proposition~\ref{pro:same-maximal-covering} because $C(\rho^{r_{n}(i)}_i)=C(\overline{\rho}_i^{r'_{n}(i)})$.
\end{proof}

\section*{Conclusions}

We have proved that the concept class $\mathcal{C}^*_U$ of finite unions of one-dimensional intersecting affine modules is efficiently exactly learnable using equivalence and subset queries and that the algorithm can be adapted to use membership queries instead of subset queries when a common element is known. The correctness and efficiency of our algorithm rely on the structure of the sets considered. In the case of a single module, the learning algorithm only uses equivalence queries. At each step, the conjectured period for the hypothesis is reduced by a prime factor, which guarantees termination. Since the number of prime factors is smaller than the base two logarithm, this gives the polynomial time upper bound. The case of the union of modules is more elaborate. An important concept in the analysis of the algorithm is the notion of maximal covering of a module. The fact that the maximal coverings of all the modules in the hypothesis are pairwise disjoint allows us to derive a polynomial bound on the number of modules in the hypothesis. We then show that each module of the hypothesis converges to a different maximal module. Since the smallest representation of a union of intersecting affine modules is given by its set of maximal modules, we conclude that $\mathcal{C}_U^*$ is efficiently learnable with the corresponding queries. Overall, our results contribute towards a better understanding of the efficient learnability of sets definable in linear integer arithmetic.

Although our work is restricted to unions of one-dimensional intersecting affine modules, it already exhibits structural features that do not appear in more general concept classes of the same kind. Thus, a natural next step is to study the learnability of unions of affine modules without the intersection requirement. More broadly, it would be interesting to investigate the learnability of related concept classes in higher dimensions. It remains open to prove that $\mathcal{C}_U^*$ is efficiently exactly learnable with membership queries instead of subset queries, even when a common element is not known.

\section*{Acknowledgments}

We thank Dmitry Chistikov, Pascual Jara, Daniel Stan and anonymous reviewers for their helpful feedback. González and Lin were partially supported by European Union\footnote{Views and opinions expressed are however those of the author(s) only and do not necessarily reflect those of the European Union or the European Research Council Executive Agency. Neither the European Union nor the granting authority can be held responsible for them.}
\includegraphics[width=0.75cm]{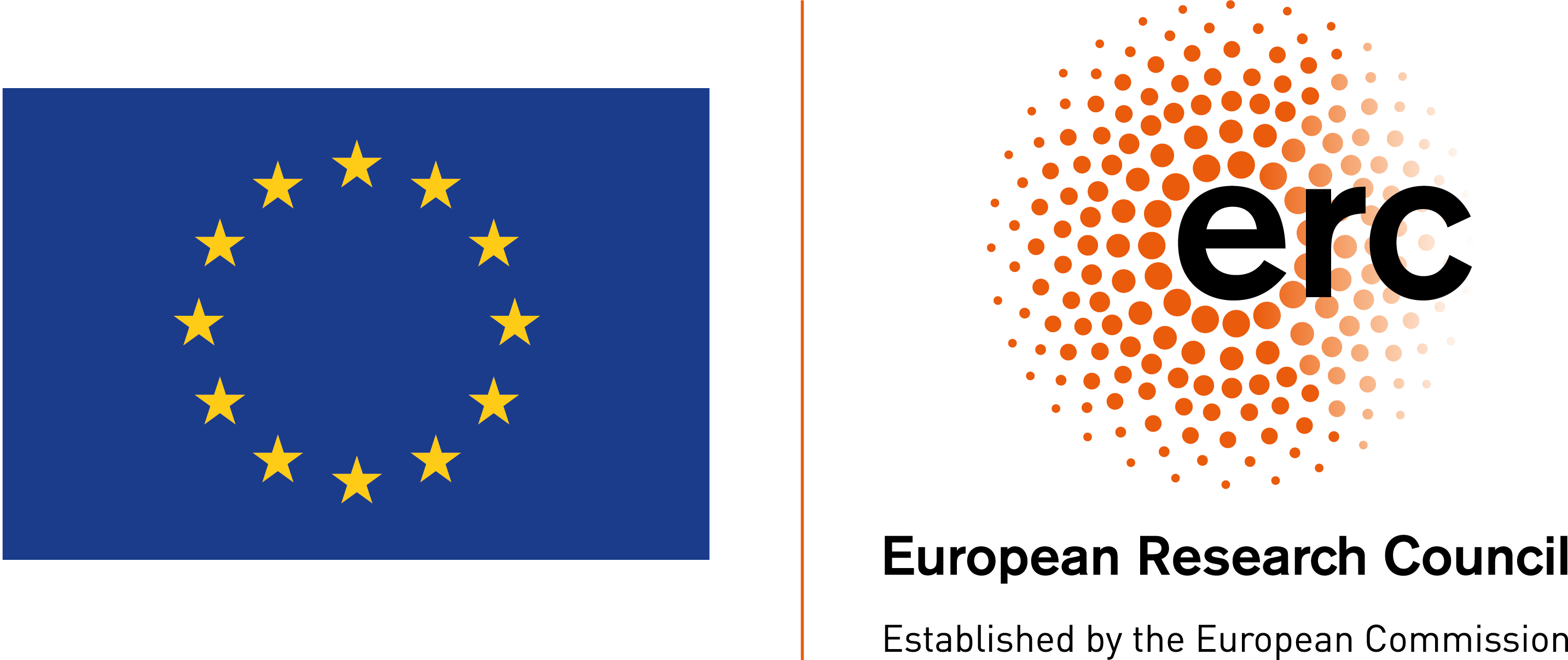} (ERC, LASD, \href{https://doi.org/10.3030/101089343}{101089343}).

\bibliographystyle{ACM-Reference-Format}
\bibliography{pmlr-sample}

@article{helmbold1992learning,
  title={Learning integer lattices},
  author={Helmbold, David and Sloan, Robert and Warmuth, Manfred K},
  journal={SIAM Journal on Computing},
  volume={21},
  number={2},
  pages={240--266},
  year={1992},
  publisher={SIAM}
}

@article{abe1995characterizing,
  title={Characterizing PAC-learnability of semilinear sets},
  author={Abe, Naoki},
  journal={Information and Computation},
  volume={116},
  number={1},
  pages={81--102},
  year={1995},
  publisher={Elsevier}
}

@article{angluin1987learning,
  title={Learning regular sets from queries and counterexamples},
  author={Angluin, Dana},
  journal={Information and computation},
  volume={75},
  number={2},
  pages={87--106},
  year={1987},
  publisher={Elsevier}
}

@book{kearns1994introduction,
  title={An introduction to computational learning theory},
  author={Kearns, Michael J and Vazirani, Umesh},
  year={1994},
  publisher={MIT press}
}

@inproceedings{markgraf2021learning,
  title={Learning Union of Integer Hypercubes with Queries: (with Applications to Monadic Decomposition)},
  author={Markgraf, Oliver and Stan, Daniel and Lin, Anthony W},
  booktitle={International Conference on Computer Aided Verification},
  pages={243--265},
  year={2021},
  organization={Springer}
}

@inproceedings{chistikov2024introduction,
  title={An introduction to the theory of linear integer arithmetic},
  author={Chistikov, Dmitry},
  booktitle={44th IARCS Annual Conference on Foundations of Software Technology and Theoretical Computer Science (FSTTCS 2024)},
  volume={323},
  pages={1},
  year={2024},
  organization={Schloss Dagstuhl—Leibniz-Zentrum f{\"u}r Informatik}
}

@article{TAKADA1992207,
title = {Learning semilinear sets from examples and via queries},
journal = {Theoretical Computer Science},
volume = {104},
number = {2},
pages = {207-233},
year = {1992},
author = {Yuji Takada}
}

@article{angluin1988queries,
  title={Queries and concept learning},
  author={Angluin, Dana},
  journal={Machine learning},
  volume={2},
  number={4},
  pages={319--342},
  year={1988},
  publisher={Springer}
}

@book{diekert2016discrete,
  title={Discrete algebraic methods: Arithmetic, cryptography, automata and groups},
  author={Diekert, Volker and Kufleitner, Manfred and Rosenberger, Gerhard and Hertrampf, Ulrich},
  year={2016},
  publisher={Walter de Gruyter GmbH \& Co KG}
}

@inproceedings{goldberg1994learning,
  title={Learning unions of boxes with membership and equivalence queries},
  author={Goldberg, Paul W and Goldman, Sally A and Mathias, H David},
  booktitle={Proceedings of the seventh annual conference on Computational learning theory},
  pages={198--207},
  year={1994}
}

@article{parikh1966context,
  title={On context-free languages},
  author={Parikh, Rohit J},
  journal={Journal of the ACM (JACM)},
  volume={13},
  number={4},
  pages={570--581},
  year={1966},
  publisher={ACM New York, NY, USA}
}

@article{littlestone1988learning,
  title={Learning quickly when irrelevant attributes abound: A new linear-threshold algorithm},
  author={Littlestone, Nick},
  journal={Machine learning},
  volume={2},
  number={4},
  pages={285--318},
  year={1988},
  publisher={Springer}
}

@article{angluin2004queries,
  title={Queries revisited},
  author={Angluin, Dana},
  journal={Theoretical Computer Science},
  volume={313},
  number={2},
  pages={175--194},
  year={2004},
  publisher={Elsevier}
}

@article{valiant1984theory,
  title={A theory of the learnable},
  author={Valiant, Leslie G},
  journal={Communications of the ACM},
  volume={27},
  number={11},
  pages={1134--1142},
  year={1984},
  publisher={ACM New York, NY, USA}
}

@book{beachy2019abstract,
  title={Abstract algebra},
  author={Beachy, John A and Blair, William D},
  year={2019},
  publisher={Waveland Press}
}

@inproceedings{kopczynski2010parikh,
  title={Parikh images of grammars: Complexity and applications},
  author={Kopczynski, Eryk and To, Anthony Widjaja},
  booktitle={2010 25th Annual IEEE Symposium on Logic in Computer Science},
  pages={80--89},
  year={2010},
  organization={IEEE}
}

@book{jones1998elementary,
  title={Elementary number theory},
  author={Jones, Gareth A and Jones, Josephine M},
  year={1998},
  publisher={Springer Science \& Business Media}
}

@inproceedings{van2020learning,
  title={Learning weighted automata over principal ideal domains},
  author={van Heerdt, Gerco and Kupke, Clemens and Rot, Jurriaan and Silva, Alexandra},
  booktitle={International Conference on Foundations of Software Science and Computation Structures},
  pages={602--621},
  year={2020},
  organization={Springer}
}

@article{buna2024learning,
  title={On learning polynomial recursive programs},
  author={Buna-Marginean, Alex and Cheval, Vincent and Shirmohammadi, Mahsa and Worrell, James},
  journal={Proceedings of the ACM on Programming Languages},
  volume={8},
  number={POPL},
  pages={1001--1027},
  year={2024},
  publisher={ACM New York, NY, USA}
}

\end{document}